\documentclass[11pt,a4paper]{article}

\usepackage{hyperref}
\usepackage{amsfonts}
\usepackage{mathrsfs}
\usepackage{comment}
\usepackage{amsmath}
\usepackage{amssymb}
\usepackage{amsthm}
\usepackage{amscd}
\usepackage[all]{xy}
\usepackage{titlesec}
\usepackage[top=20mm, bottom=20mm, left=20mm, right=20mm]{geometry}
\usepackage{color}
\DeclareMathOperator{\detect}{Detect}
\DeclareMathOperator{\correct}{Correct}

\newtheorem{theorem}{Theorem}[section]
\newtheorem{lemma}[theorem]{Lemma}
\newtheorem{proposition}[theorem]{Proposition}
\newtheorem{corollary}[theorem]{Corollary}

\newcommand{\ma}{\mathcal}

\newcommand{\s}{\subseteq}

\newcommand{\fr}{\frac}
\newcommand{\lc}{\lceil}
\newcommand{\rc}{\rceil}
\newcommand{\lf}{\lfloor}
\newcommand{\rf}{\rfloor}

\newcommand{\x}{\vec}

\begin{document}
\title{Error Detection and Correction in Communication Networks\footnote{Part of this paper will appear in  2020 IEEE International Symposium on Information Theory.}}

\author{Chong Shangguan\footnote{Department of Electrical Engineering-Systems, Tel Aviv University, Tel Aviv 6997801, Israel. Email: theoreming@163.com.} and Itzhak Tamo\footnote{Department of Electrical Engineering-Systems, Tel Aviv University, Tel Aviv 6997801, Israel. Email: zactamo@gmail.com.}}

\date{}

\maketitle

\begin{abstract}
  Let $G$ be a connected graph on $n$ vertices and $C$ be an $(n,k,d)$ code with $d\ge 2$, defined on the alphabet set $\{0,1\}^m$. Suppose that for $1\le i\le n$, the $i$-th vertex of $G$ holds an input symbol $x_i\in\{0,1\}^m$ and let $\vec{x}=(x_1,\ldots,x_n)\in\{0,1\}^{mn}$ be the input vector formed by those symbols. Assume that each vertex of $G$ can communicate with its neighbors by transmitting messages along the edges, and these vertices must decide deterministically, according to a predetermined communication protocol, that whether $\vec{x}\in C$. Then what is the minimum communication cost to solve this problem? Moreover, if $\vec{x}\not\in C$, say, there is less than $\lfloor(d-1)/2\rfloor$ input errors among the $x_i$'s, then what is the minimum communication cost for error correction?

  In this paper we initiate the study of the two problems mentioned above. For the error detection problem, we obtain two lower bounds on the communication cost as functions of $n,k,d,m$, and our bounds are tight for several graphs and codes. For the error correction problem, we design a protocol which can efficiently correct a single input error when $G$ is a cycle and $C$ is a repetition code. We also present several interesting problems for further research.
\end{abstract}



\section{Introduction}

Let $G$ be a connected simple graph on $n$ vertices $v_1,\ldots,v_n$, and $C$ a code of length  $n$ over an alphabet $Q$ of   size $2^m$,  dimension  $k:=\frac{\log_{2^m} |C|}{n}$,  and minimum Hamming distance $d\ge2$. Suppose that each vertex  $v_i$ holds (stores) an input symbol $x_i\in Q$ and let $\x{x}:=(x_1,\ldots,x_n)\in Q^n$ be the input vector  formed by the $n$ stored symbols. Assume that the vertices of $G$ have  unbounded computation power, and messages can be sent between any two neighboring vertices along the edge connecting them.

We consider the following  communication complexity type question which are very natural to consider in coding theory. Given the graph $G,$ the code $C$, and the input vector $\x{x}$, decide  in the most economic way (in terms of number of bits transmitted along the edges of $G$) that whether the input vector $\x{x}$ forms a codeword of the code $C$.

This problem is in fact the error detection problem performed in the distributed settings. Towards this end the vertices execute a  predetermined \emph{deterministic} communication protocol that outputs the correct answer with the zero error. For example,  the trivial protocol for the error detection problem would be to transmit all the input symbols  along the edges of a rooted spanning tree to the root of the tree. Upon receiving the input messages, the root can declare whether the input vector forms a codeword of the code.

Once an error is detected the next natural step would be to run an   error correction operation in order to fix those errors. Assuming at most   $\lf(d-1)/2\rf$ errors have occurred, i.e., there is a codeword of distance at most $\lf(d-1)/2\rf$ from $\x{x}$, the task is to correct the erroneous input symbols, again in the most economic way. The corresponding trivial protocol for this task is obvious.

Given $G$ and $C$, we denote the error detection and correction problems described above by $\detect(G,C)$ and $\correct(G,C)$, respectively. The corresponding communication complexity (normalized by the number of input bits $m$) are denoted by $|\detect(G,C)|$ and $|\correct(G,C)|$, respectively.

Since the trivial protocol may be rather wasteful in terms of the communication cost, the main purpose of this paper is to study the two functions $|\detect(G,C)|$ and $|\correct(G,C)|$. Our research was motivated by a recent paper of Alon, Efremenko and Sudakov \cite{nogatesting}, which studied the communication complexity of the equality testing problem, i.e., to determine whether $x_1=\cdots=x_n$. In other words, they studied the problem $\detect(G,C)$ when $C$ is a repetition code (or $Rep$ for short). Among other results, they showed that for any Hamiltonian graph $G$ and fixed integer $n\ge 3$,
\begin{equation}\label{Hamiltonian}
n/2\le |\detect(G,Rep)|\le n/2+o(1),
\end{equation}

\noindent where $o(1)\rightarrow0$ as $m\rightarrow\infty$. For a Hamiltonian graph $G$ and $m\rightarrow\infty$, clearly \eqref{Hamiltonian} determines $|\detect(G,Rep)|$ up to a lower order term.

The protocols considered in \cite{nogatesting} are all static (i.e. non-adaptive) ones such that which vertex speaks and when is determined in advance, and is independent of both inputs and transmitted messages. It is crucial that in this paper, for $\detect(G,C)$ we also focus on the static protocols. However, for $\correct(G,C)$ we do not restrict to the static situation since intuitively, to correct errors efficiently one should use the previously revealed information. The error detection and correction problems are natural generalizations of the equality testing problem, and may find further applications in the scenarios like parallel and distributed computation.

The case when $G$ is the complete graph $K_n$ is of particular interest. In this case, the trivial protocol simply sends all information to a chosen vertex, say $v_1$, so that it is able to decide whether $\x{x}\in C$. Moreover, if $\x{x}\not\in C$ and there are at most $\lf(d-1)/2\rf$ input errors, then according to the minimum distance of $C$, $v_1$ can compute the correct codeword, and then send to the corrupted vertices the correct symbols.
It follows that $|\detect(K_n,C)|\le n-1$ and $|\correct(K_n,C)|\le n-1+i$, where $i\le \lf(d-1)/2\rf$ is the maximum number of possible input errors.

Below we summarize the main contribution of this paper.
\begin{itemize}
    \item We generalize the problem studied in \cite{nogatesting} and present two lower bounds (see Theorems \ref{nk2} and \ref{general} below) for $|\detect(G,C)|$ as functions of $n,k,d$, which
    are tight for several graphs and codes. We also present a strong lower bound (see Theorem \ref{linearprotocol} below) on $|\detect(G,C)|$ for linear static protocols.

    \item Let $C_n$ denote the cycle of length $n$. By \eqref{Hamiltonian} it is clear that $n/2\le |\detect(C_n,Rep)|\le n/2+o(1)$. We improve the lower bound by showing that for fixed $n$, $m\cdot (|\detect(C_n,Rep)|-n/2)$ tends to infinity as $m$ tends to infinity (see Theorem \ref{1+triangle} below). The improvement is tiny but rather nontrivial. Indeed, our proof is based on a deep result in graph theory, namely, the celebrated graph removal lemma.

    \item Assuming that there is at most one input error, we design a protocol which solves the error correction problem efficiently, namely, we show that in this case, $|\correct(C_n,Rep)|\le n/2+1+o(1)$.
\end{itemize}


{\it Related work.} In the literature, the function $|\detect(G,Rep)|$ has been studied extensively, see, e.g. \cite{eq1}, \cite{eq2}, \cite{eq3} and \cite{eq4}. There were also many papers studying more flexible models or other problems related to communication complexity, see, e.g. \cite{other3}, \cite{eq2}, \cite{other2}, \cite{other1} and \cite{ellen2013brief}. In particular, \cite{ellen2013brief} presented a non-static protocol which solves $\detect(K_n,Rep)$ with communication cost $O(mn/\log n)$. It remains an interesting open question to study $\detect(G,C)$ for general graph $G$ and code $C$ in the non-static setting.

\section{Preliminaries}

Throughout this paper we assume that the number of vertices $n$ is fixed and $m$ can be arbitrarily large. Although we mainly consider the case when $Q=\{0,1\}^m$, our results can be easily extended to any other alphabet. For a subset $C\subseteq Q^n$ and two vectors $\vec{x}=(x_1,\ldots,x_n),\vec{y}=(y_1,\ldots,y_n)\in C$, the {\it Hamming distance} ${\rm d}(\vec{x},\vec{y})$ is defined to be the number of distinct coordinates between $\vec{x}$ and $\vec{y}$, namely, ${\rm d}(x,y)=|\{1\le i\le n:x_i\neq y_i\}|$. The {\it minimum distance} ${\rm d}(C)$ of $C$ is defined to be the minimum of ${\rm d}(\vec{x},\vec{y})$ among all pairs of distinct $\x{x},\x{y}\in C$. An $(n,k,d)$ code $C$ is simply a subset $C\subseteq Q^n$ with {\it dimension} $k:=\frac{\log_{2^m} |C|}{n}$ and minimum distance $d$.

A protocol is said to be {\it static} if it is executed in rounds, and in each round exactly one vertex sends a single message to one of its neighbors. Furthermore, the identity of the vertex transmitting the message, the receiver and the length of the message are also predetermined and do not depend on the previously transmitted messages.

We always assume that the code $C$ has minimum distance $d$ at least $2$. Let $P$ be a static protocol that solves $\detect(G,C)$. For an input vector $\vec{x}$, we call the vertices that know whether $\vec{x}\in C$ (upon executing $P$) the {\it decision vertices} of $\vec{x}$ (with respect to $P$). Note that $\vec{x}$ may have more than one decision vertices, but their decisions must be consistent as $P$ solves the error detection problem correctly. The {\it transmission history} $h_P(\vec{x})$ is an ordered binary string which records the messages transmitted in all the rounds along the edges of $G$, upon executing $P$ on the input vector  $\vec{x}$. We will occasionally  omit the subscript $P$ when it is understood from the context. We use $|\cdot|$ to denote both the size of a set and the length of a vector. Note that  since $P$ is static, $|h_P(\x{x})|=|h_P(\x{y})|$ for any $\x{x},\x{y}\in Q^{n}$.

Below we give the formal definitions of $|\detect(G,C)|$ and $|\correct(G,C)|$. Assuming that $Q=\{0,1\}^m$, then

\begin{equation*}
\begin{aligned}
|\detect(G,C)|&=\frac{1}{m}\cdot\min_{\substack{\text{static $P$ solving}\\\detect(G,C)}}~|h_P(\x{x})|.\\
  |\correct(G,C)|&=~\frac{1}{m}\cdot\min_{\substack{any~P~solving\\\correct(G,C)}}~\max_{\x{x}\in Q^n}~\{|h_P(\x{x})|\}.\\
\end{aligned}
\end{equation*}

\noindent Note that in the definition of $|\detect(G,C)|$, $\vec{x}$ can be any vector of $Q^n$ by the property of a static protocol.

The following basic definitions  from graph theory will be used in the sequel. We will only consider simple connected graphs $G=(V,E)$, with  $V$ and $E$ as the set of vertices and edges respectively.  A graph is called {\it $n$-partite} if one can color its vertices  by $n$ colors such  that it contains no edge whose two endpoints have the same color. The color classes are also called vertex parts of the $n$-partite graph. For a subset $S\s V$ of vertices let $\overline{S}$ be its complement with respect to $V$.  For two input vectors $\vec{x},~\x{y}\in Q^n$,  let $\x{x}_S\vee\x{y}_{\overline{S}}$ be the vector whose $i$ coordinate equals to $x_i$ for $i\in S$ and $y_i$ for $i\in\overline{S}$. A {\it cut} of $G$ is a partition of $V$ into two subsets $S$ and $\overline{S}$.
The {\it cut-set} of a cut is the set $(S,\overline{S}):=\{(u,v)\in E:~u\in S,~v\in\overline{S}\}$ of edges that have one endpoint in $S$ and another endpoint in $\overline{S}$.

We will frequently make use of the following lemma.

\begin{lemma}\label{beginning...}
Let $P$ be a static protocol that solves $\detect(G,C)$.
Assume that the transcripts $h_P(\x{x})$ and $h_P(\x{y})$ for  input vectors   $\vec{x},~\vec{y}\in Q^n$ are  identical  along every edge of the cut $(S,\overline{S})$, then if $P$ accepts both $\x{x}$ and $\x{y}$, then it must accept either $\vec{x}_S\vee\vec{y}_{\overline{S}}$ or $\vec{y}_S\vee\vec{x}_{\overline{S}}$.
\end{lemma}

\begin{proof}
  Let $v\in V$ be a decision vertex of $\vec{x}$.
  If $v\in S$, then $P$ will accept $\vec{x}_S\vee\vec{y}_{\overline{S}}$. Indeed, upon executing $P$,  $v$ cannot distinguish between the two input vectors $\vec{x}_S\vee\vec{y}_{\overline{S}}$ and $\vec{x}$. Similarly, if  $v\in {\overline{S}}$, then $P$ will accept $\vec{y}_S\vee\vec{x}_{\overline{S}}$, since $v$ cannot distinguish between $\vec{y}_S\vee\vec{x}_{\overline{S}}$ and $\vec{x}$.
\end{proof}

\section{Two lower bounds on the communication cost of the  error detection problem}

In this section we derive two fundamental lower bounds on the communication cost of the  error detection problem which are a function of basic parameters of the code, $n,k,d$. For large $k$ the first bound is tighter than the second one, whereas for large $d$, the second bound is tighter. Then we proceed to give a strong lower bound on the communication complexity for the important sub-family of linear protocols.

The next theorem gives the first lower which claims that the communication complexity of the  error detection is at least the dimension of the underlying code.

\begin{theorem}\label{nk2}
Let   $G$ be a connected graph and  $C$ an $(n,k,d)$ code, then  $|\detect(G,C)|\ge k$.
\end{theorem}

\begin{proof}
If there exists some protocol that solves $\detect(G,C)$ by transmitting less than  $km$ bits, then by the pigeonhole principle there exist two distinct input vectors  $\vec{x},~\vec{y}\in C$  with the same transcript, i.e.,  $h(\vec{x})=h(\vec{y})$.
Assume without loss of generality that $x_1\neq y_1$, and  let $\vec{z}=(y_1,x_2,\ldots,x_n)$ and $\vec{w}=(x_1,y_2,\ldots,y_n)$. Then, by Lemma \ref{beginning...}, the protocol must accept either $\vec{z}$ or $\x{w}$, which implies that $\vec{z}\in C$ or $\x{w}\in C$, thereby contradicting the minimum distance of $C$ (we assumed that the minimum distance is always at least 2).
\end{proof}

Although Theorem \ref{nk2} has an almost trivial proof, it provides a relatively  good bound for codes with large dimension, which  in some cases  is even tight.  Indeed, assume for  instance that $C$  is the $(n,n-1,2)$  parity check code, with $Q$ being the finite field $\mathbb{F}_{2^m}$, i.e., it consists of all  the vectors   $\x{x}\in(\mathbb{F}_{2^m})^n$ that satisfy $\sum_{i=1}^n x_i=0$. Then, by Theorem \ref{nk2} the normalized communication cost is at least $n-1$. Note that this bound is indeed tight,   as a matching upper bound is  provided by the following protocol. Fix a rooted spanning tree of $G$, and let each leaf sends its symbol to its unique parent. Each other vertex which is not a leaf nor the root, upon receiving all the messages from its children,  XORs them together with its own input symbol and sends the outcome to its parent.  Clearly, the root of the tree upon receiving all the messages from its children can decide whether $\x{x}$ is a codeword or not, and exactly one symbol of $\mathbb{F}_{2^m}$ is being transmitted along each of the $n-1$ edges of the tree. Hence we have the following corollary.

\begin{proposition}\label{nn-12}
Let $C$ be the $(n,n-1,2)$ parity check code, then for any connected graph $G$, $|\detect(G,C)|=n-1$.
\end{proposition}

We proceed to the next fundamental lower bound which applies for code with large minimum distance. In particular, it provides tighter bounds than Theorem \ref{nk2} for codes with relatively small dimension but with large minimum distance. It in fact extends Theorem 4 of \cite{nogatesting} to any code with large minimum distance.

\begin{theorem}\label{general}
    Let $C$ be an $(n,k,d)$ code
    with $n\le 2(d-1)$, then  $|\detect(G,C)|$ is at least the solution to the following  LP problem:
    \begin{equation*}
      \begin{aligned}
        &maximize~k\cdot \sum_{\substack{(S,\overline{S})\in\ma{S}\\|S|=n-d+1}} g(S,\overline{S}),\\
        &subject~to~\sum_{\substack{(S,\overline{S})\in\ma{S}_e\\|S|=n-d+1}}g(S,\overline{S})\le 1~for~every~e\in E,\\
      \end{aligned}
    \end{equation*}
\noindent where $\ma{S}$ is the family of all cut-sets of $G$,  $\ma{S}_e$ is the family of cut-sets that contain the edge $e$, and $g$ is a function that assigns non-negative weights to all cut-sets $(S,\overline{S})\in\ma{S}$ with $|S|=n-d+1$.
\end{theorem}

Before we proceed to prove Theorem \ref{general} we give two results that follow from it.

\begin{corollary}\label{nkd}
  Let 
  $C$ be an $(n,k,d)$ code with  $n\le 2(d-1)$. Then  $|\detect(G,C)|\ge\fr{kn(n-1)}{2(n-d+1)(d-1)}.$
\end{corollary}

\begin{proof}[Sketch of the proof]
Since $|\{(S,\overline{S})\in\ma{S}_e:|S|=n-d+1\}|\le 2\binom{n-2}{n-d}$ for any $e\in E$,  the corollary follows by setting $g(S,\overline{S})=(2\binom{n-2}{n-d})^{-1}$ to all $(S,\overline{S})$ with $|S|=n-d+1$.
\end{proof}

The following result is  Corollary \ref{nkd} applied to MDS codes, i.e., codes  with minimum distance  $d=n-k+1$.

\begin{corollary}
  \label{mdsnkd}
 Let $G$ be a  connected graph  and $C$ be an  $(n,k)$ MDS code with $n\ge 2k$, then  $|\detect(G,C)|\ge\fr{n(n-1)}{2(n-k)}$.
\end{corollary}

Corollary \ref{mdsnkd} implies that if $C$ is an $(n,n/2)$ MDS code, then even on the complete graph $K_n$, the normalized communication cost of solving $\detect(K_n,C)$ is at least $(n-1)$, which is no better than the trivial protocol.
It would be interesting to determine whether one can solve $\detect(K_n,C)$ nontrivially for any $(n,k)$ MDS code with $k\ge 2$.

Combining Theorem \ref{nk2} and Corollary \ref{mdsnkd}, we have the following result.

\begin{corollary}
  Let $G$ be a  connected graph  and $C$ be an  $(n,k)$ MDS code. Then  $|\detect(G,C)|\ge\max\{k,\fr{n(n-1)}{2(n-k)}\}$.
\end{corollary}

We proceed to prove  Theorem \ref{general}.

\begin{proof}[\textbf{Proof of Theorem \ref{general}}]
Let $(S,\overline{S})$ be a cut-set of $G$ with $|S|=n-d+1$. We claim that the total amount of bits transmitted along the edges of $(S,\overline{S})$ during the execution of an error detection protocol is at least  $km$.
Assume  otherwise, then by the pigeonhole principle there exist two  distinct codewords $\x{x},~\x{y}\in C$ such that $\x{x}$ and $\x{y}$ have identical transcripts along each edge of $(S,\overline{S})$. Then, by Lemma \ref{beginning...} the protocol  must accept either $\vec{x}_S\vee\vec{y}_{\overline{S}}$ or $\vec{y}_S\vee\vec{x}_{\overline{S}}$. Assume without loss of generality that  $\vec{x}_S\vee\vec{y}_{\overline{S}}$ is accepted by the protocol, which means that this is also a codeword of $C$. However, the Hamming distance between $\x{x}$ and $\vec{x}_S\vee\vec{y}_{\overline{S}}$ is at most $|\overline{S}|=d-1$, hence $\x{x}=\vec{x}_S\vee\vec{y}_{\overline{S}}$.
Then, the Hamming distance of $\x{x}$ and $\x{y}$ is at most $|S|=n-d+1\le d-1$,
and we reach to a contradiction.


Let $t(e)$ be the normalized number of bits transmitted along the edge $e\in E$, then by the above bound on the cuts we get that the solution of the following linear program provides a lower bound for $|\detect(G,C)|$.

\begin{equation*}
      \begin{aligned}
        &minimize~\sum_{e\in E} t(e),\\
        &subject~to~ 
        \sum_{e\in(S,\overline{S})} t(e)\ge k~for~every~(S,\overline{S})\in\ma{S}\\
        &~with~|S|=n-d+1,~and~t(e)\ge0~for~every~e\in E.\\
      \end{aligned}
    \end{equation*}
\noindent The theorem follows from the duality of LP.
\end{proof}

Next we restrict our attention to  linear protocols which  probably form the most important sub-family of protocols.
A protocol over the alphabet $Q=\mathbb{F}_2^m$  is said to be {\it linear} if any transmitted bit by a vertex during the execution of the protocol is a linear function of the input symbol of the vertex and its previously received messages (bits).

The paper \cite{nogatesting} already considered this problem for the case of detecting an error for the repetition code over Hamiltonian graphs, where it was shown that linear protocols are strictly sub-optimal than the most efficient protocol, which is clearly non-linear. More precisely, \cite{nogatesting} showed that if $C$ is the $(n,1,n)$ repetition code, then
$|\detect(G,C)|=n-1$, where the upper bound is attained by the trivial protocol.

Next  we show that the same result holds for arbitrary  non-trivial codes, i.e., codes with minimum distance at least $2$. We emphasize  that the   result is a nontrivial generalization of \cite{nogatesting}.

 \begin{theorem}\label{linearprotocol}
 Let 
 $C$ be an $(n,k,d)$ code, then the normalized communication cost for any linear static protocol that solves $\detect(G,C)$ is at least $(n-1)$.
 \end{theorem}

We will need  the following lemma, which shows that there is only a  small number of input vectors that have a transmission history that equals the transmission history of a specific codeword.


 \begin{lemma}\label{sharehistory}
   For each $\x{x}\in C$ there are at most   $2^m$ input vectors $\x{y}\in Q^n$ with $h(\x{y})=h(\x{x})$.
 \end{lemma}

 \begin{proof}
By a similar argument as in the proof of Theorem \ref{nk2}, it is not hard to see that  $h(\x{x})\neq h(\x{y})$ for $\x{y}\in C\setminus\{\x{x}\}$. Hence,
let $\x{y}\in Q^n \setminus C$ be a vector that satisfies $h(\x{x})=h(\x{y})$ and assume without loss of generality that $v_1$ is a decision vertex of $\x{x}$.
Then, it is clear that $x_1\neq y_1$ since otherwise the protocol would accept also $\x{y}$ at $v_1$.
We claim that $y_1$ is the only coordinate in which $\x{y}$ differs from $\x{x}$.
Assume towards  contradiction that $y_j\neq x_j$ for some $j\neq 1$.
Then the protocol must accept $\x{x}_{V\setminus\{v_j\}}\vee\x{y}_{\{v_j\}}$ at $v_1$, since $v_1$ cannot distinguish between it and  $\x{x}$.
   But then,  $\x{x}_{V\setminus\{v_j\}}\vee\x{y}_{\{v_j\}}\in C$, which is impossible as $d\ge 2$.
   Thus there are at most $2^m$ (including $\x{x}$)  input vectors $\x{y}\in Q^n$ with $h(\x{x})=h(\x{y})$, as needed.
 \end{proof}

 Next we  present the proof of Theorem \ref{linearprotocol}.

 \begin{proof}[\textbf{Proof of Theorem \ref{linearprotocol}}]
  Consider any linear static protocol that solves $\detect(G,C)$, and suppose that during its execution under some  input vector $\x{x}\in C$, $s$ bits are transmitted, then it suffices to show that $s\geq (n-1)m.$

  Let  $h(\x{x})=(h_1,\ldots,h_s)\in\mathbb{F}_2^s$, where each $h_i$ is a bit. Since the protocol is linear,  each  $h_i$ is a linear combination of the $mn$ input bits of the vector $\x{x}$.
 The protocol thus can be represented by  $s$ linear functions $g_i:\mathbb{F}_2^{mn}\longrightarrow\mathbb{F}_2$ such that $g_i(\x{x})=h_i$ holds for every $1\le i\le s$.
   Next, consider the system of linear equations $$g_i(\x{z})=h_i,\qquad 1\le i\le s.$$
   By definition, such a system has a valid solution $\x{z}=\x{x}$. Then it must contain at least $2^{mn-s}$ distinct solutions. Equivalently,  there are at least $2^{mn-s}$ distinct input vectors that have the same transmission history with $\x{x}$. By Lemma \ref{sharehistory} we conclude that $mn-s\le m$, as needed.
 \end{proof}

\section{Error detection for the repetition code - An improved lower bound}\label{removal}
In this section we give an improvement on the lower bound \eqref{Hamiltonian} which originally was proved  in \cite{nogatesting}.

Let $Rep$ be the $(n,1,n)$ repetition code and $C_n$ be the $n$-cycle with vertices $v_i$ and edges $e_j$ written in the order  $(v_1,e_1,v_2,\ldots,v_n,e_n,v_1)$.
We will show that the normalized communication cost of any static protocol that solves $\detect(C_n,Rep)$ is at least $n/2+s(m)$, where $s(m)$ is a function such that $m\cdot s(m)$ tends to infinity with $m$.

Let
$F$ be an $n$-partite graph with vertex parts $I_1,\ldots,I_n$.
As in \cite{nogatesting}, a {\it special $n$-cycle} is a copy of $C_n$ in $F$ such that for each $i$, the vertex playing the role of $v_i$ belongs to $I_i$.
We consider $n$-partite graphs that satisfy the following properties:

\begin{itemize}
     \item [$(1)$] it contains exactly $2^m$ edge disjoint special $n$-cycles;
     \item [$(2)$] each edge is contained in exactly one special $n$-cycle;
     \item [$(3)$] it contains exactly $2^m$ special $n$-cycles.
\end{itemize}
Note that it is clear that $(3)$ is a consequence of $(1)$ and $(2)$.

The following theorem provides a necessary condition for any static protocol that solves $\detect(C_n,Rep)$.

\begin{theorem}\label{n-cycle-repetition}
Any static protocol that solves $\detect(C_n,Rep)$ induces an $n$-partite graph $F$ with vertex parts $I_1,\ldots,I_n$ that satisfy $(1)$ and $(2)$.
Moreover, the protocol requires the transmission of at  least $\sum_{j=1}^n\lc\log|I_j|\rc$ bits.
\end{theorem}

\begin{proof}
Fix a protocol that solves $\detect(C_n,Rep)$.
For an input vector $\x{x}\in Rep$ and an edge $e_j$, let $j(\x{x})$ denote the messages transmitted along $e_j$ on the input $\x{x}$ during the execution of the protocol.

Define $I_j$, the $j$-th part of $F$ to be   $I_j=\{j(\x{x}):\x{x}\in Rep\}$. Next, we define the edges of $F$, which are in fact defined by the special $n$-cycles as follows.
For  $x\in Q,$ define the special $n$-cycle  $C_n^{(x)}$ by the $n$ edges  $(1(\x{x}),2(\x{x})),\ldots,(n(\x{x}),1(\x{x}))$, and let the edges of $F$ be  the union  of these $2^m$ special $n$-cycles.
We proceed to show that $F$ satisfies $(1)$ and $(2)$.

By construction each edge of $F$ is contained in at least one special $n$-cycle.
Thus it remains to show that: $(a)$ the $2^m$ special $n$-cycles $C_n^{(x)}$ are pairwise edge disjoint; $(b)$ $\{C_n^{(x)}:x\in Q\}$ are the only special cycles contained in $F$.

To prove $(a)$, assume by contradiction that there exist distinct $x,~y\in Q$ such that $C_n^{(x)}$ and $C_n^{(y)}$ have a common edge whose endpoints are in vertex parts, say  $I_1$ and $I_2$.
It follows that $1(\x{x})=1(\x{y})$ and $2(\x{x})=2(\x{y})$, that is, for input vectors  $\x{x},~\x{y}\in C$,
the messages transmitted along the edges $e_1$ and $e_2$ are identical.
Since $\{e_1,e_2\}$ forms a cut-set for the vertex partition  $(\{v_2\},V\setminus\{v_2\})$ (where $V=\{v_1,\ldots,v_n\}$), then by Lemma \ref{beginning...} the protocol accepts either $\x{x}_{\{v_2\}}\vee\x{y}_{V\setminus\{v_2\}}$ or $\x{y}_{\{v_2\}}\vee\x{x}_{V\setminus\{v_2\}}$, which is impossible.

To prove $(b)$, assume that  $(1(\x{x}_1),\dots, n(\x{x}_n))$ form a special $n$-cycle for not  necessarily distinct codewords   $\x{x}_1,\ldots,\x{x}_n\in Rep$. 
Since $\{e_j,e_{j+1}\}$ form a cut-set for $(\{v_{j+1}\},V\setminus\{v_{j+1}\})$, then similar to the proof of $(a)$, by Lemma \ref{beginning...}
one can show that $\x{x}_{j+1}$ is uniquely determined by $j(\x{x}_j)$ and $j+1(\x{x}_{j+1})$ (that is, for distinct $\x{x},~\x{y}\in Rep$, the messages  transmitted along the cut-set $\{e_j,e_{j+1}\}$ cannot be identical).

Consider the input vector $\x{x}:=(x_1,\ldots,x_n)$, where $x_i\in Q$ is the unique symbol that defines the codeword $\x{x}_i$.
Observe that for each $1\le j\le n$, the protocol sends the same messages along $e_j$ and $e_{j+1}$ given on two different input vectors  $\x{x}$ and $\x{x}_j$.
Consequently, $v_j$ (for any $j$) cannot distinguish between $\x{x}_j$ and $\x{x}$, and hence the protocol must also  accepts $\x{x}$ as well. Since the protocol only accepts codewords of $Rep$, we conclude that $\x{x}\in Rep$ and the  $n$-cycle $(1(\x{x}_1),\dots, n(\x{x}_n))$ must be $C_n^{(x)}$ for some $x\in Q$.

Since the protocol is static, it follows by definition that $|j(\x{x})|=|j(\x{y})|$ for any $1\le j\le n$ and $\x{x},~\x{y}\in Q^n$, which further implies that $|j(\x{x})|\ge \lc\log |I_j|\rc$.
The lower bound on the total number of bits transmitted follows from the fact that  $|h(\x{x})|=\sum_{j=1}^n |j(\x{x})|$.
\end{proof}
Next we present a new lower bound for $|\detect(C_n,Rep)|$.

\begin{theorem}\label{1+triangle}
  There exists a function $s(m)$ such that $m\cdot s(m)$ tends to infinity with $m$ and for fixed $n\ge 3$, $|\detect(C_n,Rep)|\ge n/2+s(m)$.  
\end{theorem}

To prove the theorem we will make use of the following deep result in graph theory, which can be viewed as a colored version of the celebrated graph removal lemma (see, e.g. \cite{newremoval}).

 \begin{lemma}[see \cite{newremoval}, Section 5, the colored removal lemma]\label{oneedgeonetriangle}
 For every fixed $n$ there exists an $N_0>0$ such that the following statement holds.
 Any $n$-partite graph $F$ on $N>N_0$ vertices that satisfies $(1)$ and $(2)$ contains at most $\epsilon N^2$ special $n$-cycles, where more precisely $\epsilon=\epsilon(N)$ satisfies $\log^*N<1$ and $\log^*(\cdot)$ is the iteration of $\log(\cdot)$ for $O(\log \epsilon^{-1})$ times.
\end{lemma}


\begin{proof}[\textbf{Proof of Theorem \ref{1+triangle}}]
  Using Theorem \ref{n-cycle-repetition}, to prove Theorem \ref{1+triangle} it suffices to prove a lower bound on $\sum_{j=1}^n \log |I_j|$.

  We claim that $\log |I_j|+\log |I_l|\ge m$ holds for any distinct $j,l\in\{1,\ldots,n\}$.
  Indeed, since any two distinct edges of $C_n$ form a cut-set, similar to the proof of Theorem \ref{nk2} (using Lemma \ref{beginning...} and the pigeonhole principle) it is not hard to see that there are at least $2^m$ distinct transcripts transmitted along those two edges, namely, for any two edges $e_j,e_l$, $$|\{(j(\x{x}),l(\x{x})):\x{x}\in Rep\}|\ge 2^m.$$
  It follows that $|I_j||I_l|\ge 2^m$, completing the proof of the claim.

  Since $F$ satisfies $(1)$ and $(2)$, it follows by Lemma \ref{oneedgeonetriangle} that $2^m=\epsilon(N)N^2$, where $N=\sum_{j=1}^n|I_j|$ is the number of vertices of $F$.
  Hence there exists a function $g(m)$ such that $N=2^{\fr{m}{2}}g(m)$ and $g(m)\rightarrow\infty$ as $m\rightarrow\infty$.
  It is also clear that $\log N=\fr{m}{2}+\log g(m)$. 

  Let $j_1,\ldots,j_n$ be a permutation of $1,\ldots,n$ so that $|I_{j_1}|\le\cdots\le |I_{j_n}|$.
  Then it is clear that $N=\Theta(|I_{j_n}|)$ and $\log|I_{j_n}|\ge\fr{m}{2}+\log g(m)+c$ for some constant $c$ (which possibly depends on $n$).
  By the claim above, it is not hard to see that $\sum_{l=1}^{n-1} \log |I_{j_l}|\ge (n-1)m/2$.
  We conclude that $$\sum_{j=1}^n \log |I_{j}|= \sum_{l=1}^{n-1} \log |I_{j_l}|+\log|I_{j_n}|\ge nm/2+w(m)$$ holds for some $w(m)$ satisfying $w(m)\rightarrow\infty$ as $m\rightarrow\infty$.

 We can complete the proof of Theorem \ref{1+triangle} by setting $s(m)=\fr{w(m)}{m}$.
\end{proof}

\section{Efficiently correcting a single input error in repetition codes over Hamiltonian graphs}

In this section we study the error correction problem over the graph $C_n$ and the code $Rep$.
We consider only the case that there is at most one input error among the vertices of $C_n$. The goal is  to design a protocol such that
\begin{itemize}
    \item it solves $\detect(C_n,Rep)$ correctly;
    \item if an error is declared in the error detection stage, then it must correct this error efficiently.
\end{itemize}
It is an interesting open problem to study whether one can correct more than one errors efficiently.
This protocol is non-static and works also for any Hamiltonian graph.

\begin{theorem}\label{repair}
  Let $F$ be an $n$-partite graph with vertex parts $I_1,\ldots,I_n$, which satisfies $(1)$ and $(2)$ mentioned in Section \ref{removal}. Then assuming that there is at most one input error among the vertices of $C_n$, based on $F$ we can design a  protocol which solves $\correct(C_n,Rep)$ with communication cost at most $\sum_{i=1}^n \lc\log|V_i|\rc+2\cdot\max_{1\le i\le n}\lc\log|V_i|\rc$ bits.
\end{theorem}

Given the correctness of Theorem \ref{repair}, it suffices to construct a graph $F$ that satisfies $(1)$ and $(2)$. To this end, we will make use of a construction presented in \cite{nogatesting}.

\begin{lemma}[see Lemma 9 \cite{nogatesting}]\label{construction}
For fixed $n$ and sufficiently large $m$, there exists an $n$-partite graph with vertex parts $V_1,\ldots,V_n$ that satisfies $(1)$ and $(2)$. Moreover, for each $1\le i\le n$, $|V_i|=\Theta(2^{m/2-o(m)})$.
\end{lemma}

The following corollary is straightforward.

\begin{corollary}\label{repaircycle}
Assuming that there is at most one input error, then $|\correct(C_n,Rep)|\le n/2+1+o(1)$ (which costs only 1 more normalized bit than the error detection problem).
\end{corollary}

Next we present the proof of Theorem \ref{repair}.
\begin{proof}[\textbf{Proof of Theorem \ref{repair}}] According to properties $(1),(2)$ and $(3)$, $F$ contains exactly $2^m$ special $n$-cycles and those cycles are pairwise edge disjoint. Since $|Q|=2^m$, we may label each of the cycles by an element $x\in Q$, and then denote it by $C_n^{(x)}$. It is clear that for $1\le i\le n$, any $v\in V_i$ can be represented by a binary string of length $\lc\log |V_i|\rc$.

The protocol is designed as follows. Fix an arbitrary input $\x{x}=(x_1,\ldots,x_n)\in Q^n$ of $C_n$, where, as usual, the vertices of $C_n$ are labeled by the sequence $(v_1,\ldots,v_n)$. Note that the arithmetic on the subscript is assumed to be modulo $n$.

\vspace{5pt}

\noindent\textbf{Error detection stage:}


\noindent{\it Step $(i)$.} For $1\le i\le n$, when $v_i$ receives its input $x_i$, it computes $u_i:=C_n^{(x_i)}\cap V_i$ and sends $u_i$ to $v_{i+1}$, where for each $i$, $u_i$ is in fact a binary string of length $\lc\log |V_i|\rc$.
\vspace{5pt}

\noindent{\it Step $(ii)$.} After $v_{i+1}$ receives $u_i$, it computes $C_n^{(x_{i+1})}\cap V_i$ and then checks whether $u_i=C_n^{(x_{i+1})}\cap V_i$.
\vspace{5pt}

\noindent{\it Step $(iii)$.} The protocol accepts $\x{x}\in Rep$ if there is no inequality in Step $(ii)$; otherwise it rejects $\x{x}$.

\vspace{5pt}

It is not hard to verify that the $n$ equalities in Step $(ii)$ hold simultaneously if and only if $u_1,\ldots,u_n$ form a special $n$-cycle $C_n^{(*)}$, and crucially, such a cycle must share a common edge with each $C_n^{(x_i)}$. It thus follows by the construction of $F$ that $C_n^{(*)}=C_n^{(x_1)}=\cdots=C_n^{(x_n)}$ and $x_1=\cdots=x_n$. Therefore the error detection stage of the protocol solves $Detect(C_n,Rep)$ correctly.

If $\x{x}\not\in Rep$ and (by assumption) there is at most one input error among the vertices of $C_n$, then the protocol detects at least one and at most two inequalities in Step $(ii)$. Moreover, if there are in fact two inequalities, then they must occur at two consecutive vertices. The protocol invokes the following error correction stage.

\vspace{5pt}

\noindent\textbf{Error correction stage:}

\noindent{\it Step $(iv\text{-}a)$.} If the protocol detects two inequalities in Step $(ii)$, say,
$u_{j-1}\neq C_n^{(x_{j})}\cap V_{j}$ (at vertex $v_{j}$) and $u_{j}\neq C_n^{(x_{j+1})}\cap V_{j+1}$ (at vertex $v_{j+1}$).
Then it makes $v_{j-1}$ send to $v_j$ the vertex $C_n^{(x_{j-1})}\cap V_j$.
\vspace{5pt}

\noindent{\it Step $(iv\text{-}b)$.} If the protocol detects only one inequality in Step $(ii)$, say,
$u_{j}\neq C_n^{(x_{j+1})}\cap V_{j+1}$ (at vertex $v_{j+1}$). Then it makes $v_{j-1}$ send to $v_j$ the vertex $C_n^{(x_{j-1})}\cap V_j$, and $v_{j+2}$ send to $v_{j+1}$ the vertex $C_n^{(x_{j+2})}\cap V_{j+2}$.

\vspace{5pt}

\noindent\textbf{Verification:}

    Next we verify the validity of the error correction stage. According to the discussion above, 
    if the protocol invokes Step $(iv\text{-}a)$, then $\{v_{j-1},v_j\}$ must contain a corrupted vertex, and so does $\{v_j,v_{j+1}\}$.
    It follows that $v_j$ must be the only vertex that holds the erroneous input, which implies that there is $x\in Q$ such that $x_{j}\neq x$ and $x_i=x$ for any $i\in\{1,\ldots,n\}\setminus\{j\}$. Since $v_j$ knows the value of $C_n^{(x)}\cap V_{j-1}$ (by Step $(i)$) and $C_n^{(x)}\cap V_{j}$ (by Step $(iv\text{-}a)$), it knows an edge of $C_n^{(x)}$. Thus by the property of $F$, $v_j$ is able to compute $C_n^{(x)}$ and hence $x$, so the protocol corrects $x_j$ successfully.

    If the protocol invokes Step $(iv\text{-}b)$, then the unique corrupted vertex must be either $v_{j}$ or $v_{j+1}$. Therefore there exists $x\in Q$ such that $x_i=x$ for any $i\in\{1,\ldots,n\}\setminus\{j,j+1\}$. Similar to the discussion in the previous case, $v_j$ is able to compute $x$ correctly. In the meanwhile, $v_{j+1}$ knows the value of $C_n^{(x)}\cap V_{j+1}$ (as in Step $(ii)$ vertex $v_{j+2}$ detects an equality) and $C_n^{(x)}\cap V_{j+2}$ (by Step $(iv\text{-}b)$), it is also able to compute $x$ correctly.

    Finally, the upper bound on the communication cost follows easily from the construction of the protocol.
\end{proof}

\section{Correcting a single input error in repetition codes over a triangle - An example}

In this section we present a protocol that corrects a single input error in repetition codes over a triangle with communication cost $2.5m+o(m)$ bits. Our protocol is based on the error detection protocol (that solves $\detect(C_3,Rep)$) introduced in \cite{nogatesting}.

For an integer $N$, a subset $M\s\{1,\ldots,N\}$ is said to be 3-term arithmetic progression-free (or 3-AP-free for short) if it contains no 3-APs, that is, for any $\alpha,~\beta,~\gamma\in M$, the equality $\alpha+\beta=2\gamma$ holds if and only if $\alpha=\beta=\gamma$.
A well-known construction of Behrend \cite{Behrend} showed that for sufficiently large $N$, there exists 3-AP-free set $M$ with size $A(N):=\Omega(\frac{N}{2^{O(\sqrt{log N})}})$.

For sufficiently large $m$ and a given alphabet $Q=\{0,1\}^m$, let $N$ be the smallest integer such that $NA(N)\ge 2^m$. Then it is sufficient to pick $N=2^{0.5m+o(m)}$. As a consequence, any integer in $\{1,\ldots,3N\}$ can be represented by a binary string of length $0.5m+o(m)$. Fix any injective mapping $$T:~Q\longrightarrow\{(\alpha,\alpha+\beta,\alpha+2\beta):~\alpha\in\{1,\ldots,N\},~\beta\in M\},$$ so that each $x\in Q$ is uniquely determined by the triplet $T(x)$.

Let $Rep$ be the $(3,1,3)$ repetition code defined on $Q$.
We write the vertices and edges of $C_3$ in an anticlockwise sequence $(v_1,e_1,v_2,e_2,v_3,e_3,v_1)$.
Next we simulate the error detection and correction stages described in Theorem \ref{repair}.

Let $\x{x}=(x_1,x_2,x_3)\in Q^3$ be an input vector of $C_3$ such that for $1\le i\le 3$, $v_i$ holds the input symbol $x_i$. Assume further that there is at most one input error among the vertices.
The protocol works as follows. Note that the arithmetic on the subscript is assumed to be modulo 3.

\vspace{5pt}

\noindent\textbf{Error detection stage:}

\noindent{\it Step $(i)$.} For $1\le i\le 3$, when $v_i$ receives $x_i$, it computes the triplet $T(x_i)=(\alpha_i,\alpha_i+\beta_i,\alpha_i+2\beta_i)$, and sends $\alpha_i+(i-1)\beta_i$ to $v_{i+1}$ through $e_i$.

\vspace{5pt}

\noindent{\it Step $(ii)$.} $v_{i+1}$ checks whether $\alpha_i+(i-1)\beta_i=\alpha_{i+1}+(i-1)\beta_{i+1}$.

\vspace{5pt}

\noindent{\it Step $(iii)$.} The protocol accepts $\x{x}\in Rep$ if there is no inequality in Step $(ii)$; otherwise it rejects $\x{x}$.

\vspace{5pt}

\noindent\textbf{Error correction stage:}

\vspace{5pt}

\noindent{\it Step $(iv\text{-}a)$.} If the protocol detects two inequalities in Step $(ii)$, for example,  $\alpha_3+2\beta_3\neq\alpha_1+2\beta_1$ at vertex $v_1$ and $\alpha_1\neq\alpha_2$ at vertex $v_2$. Then it makes $v_{3}$ send $\alpha_{3}$ to $v_1$.

\vspace{5pt}

\noindent{\it Step $(iv\text{-}b)$.} If the protocol detects only one inequality in Step $(ii)$, for example, $\alpha_1\neq\alpha_2$ at vertex $v_2$. Then it makes $v_3$ send $\alpha_3$ to $v_1$ and $\alpha_3+2\beta_3$ to $v_2$.

\vspace{5pt}

\noindent\textbf{Verification:}

First of all, we claim that $\x{x}\in Rep$ if and only if in Step $(ii)$ the protocol detects three equalities, that is, the following three equations hold simultaneously:
  \begin{equation*}
  \left\{
    \begin{aligned}
        \alpha_1&=\alpha_{2}~(at~v_2)&\\
        \alpha_2+\beta_2&=\alpha_{3}+\beta_{3}~(at~v_3)&\\
        \alpha_3+2\beta_3&=\alpha_{1}+2\beta_{1}~(at~v_1)&\\
    \end{aligned}
  \right.
  \end{equation*}
  Adding both hand sides of the equations, it follows that $$\beta_2+\beta_3=2\beta_1.$$
  Since $M$ is 3-AP-free, the equality above holds if and only if $\beta_1=\beta_2=\beta_3$, which further implies that $\alpha_1=\alpha_2=\alpha_3$ and $T(x_1)=T(x_2)=T(x_3)$. We conclude that $x_1=x_2=x_3$ and $\x{x}\in Rep$, i.e., the protocol solves $Detect(C_3,Rep)$ correctly (we remark that the proof above is in fact a special case of the error detecting protocol of \cite{nogatesting}).

  Now it remains to consider the case when there is a single input error among the vertices of $C_3$. If the protocol invokes Step $(iv\text{-}a)$, then since both $\{v_1,v_3\}$ and $\{v_2,v_3\}$ contain a corrupted vertex, it is easy to see that such a vertex must be $v_1$. We may assume that there is an element $x\in Q$ such that $x_2=x_3=x$ and $T(x)=(\alpha,\alpha+\beta,\alpha+2\beta)$.
  Moreover, since $v_1$ knows the value of $\alpha+2\beta$ (by Step $(i)$) and $\alpha$ (by Step $(iv\text{-}a)$), it is able to compute $(\alpha,\beta)$ and hence $x$, as needed.

  If the protocol invokes Step $(iv\text{-}b)$, then unique corrupted vertex must be either $v_1$ or $v_2$. We conclude that $v_3$ must hold the correct input symbol. By similar reasoning it is not hard to show that both $v_1$ and $v_2$ are able to recover the symbol held by $v_3$.

  Since any $\alpha,~\beta\in\{1,\ldots,N\}\}$ can be represented by a binary string of length $0.5m+o(m)$, it is easy to see that the total communication cost is at most $2.5m+o(m)$ bits, as needed.


\section{Concluding remarks and open problems}
We initiate the study of
the communication complexity of error detection and correction in communication networks.
There are many interesting problems that remain widely open. 

\begin{itemize}
    \item [(1)] Study the lower and upper bounds of $|\detect(G,C)|$ for general graph $G$ and code $C$ in the non-static (i.e., adaptive) setting.

    \item [(2)] For which code $C$ does there exist a static protocol that solves $\detect(K_n,C)$ nontrivially?
 
    \item [(3)] More precisely, let $C$ be an $(n,2)$ Reed-Solomon code whose evaluation points are known to all vertices of $K_n$. Does there exist a static protocol that solves $\detect(K_n,C)$ nontrivially?
    This problem is essentially equivalent to determining whether the $n$ inputs symbols of $K_n$ are collinear (over $\mathbb{F}_q^2$ for some finite field $\mathbb{F}_q$).
    
    \item [(4)] Can we prove an information theoretic lower bound for $|\correct(G,C)|$? Does there exist a general algorithm that allows us to efficiently correct up to $\lf(d-1)/2\rf$ input errors in any $(n,k,d)$ code over $K_n$, say, the $(n,k)$ MDS codes?
\end{itemize}

Note that for any code $C$, a trivial static protocol that solves $\detect(K_n,C)$ requires $(n-1)m$ bits. Moreover, for the third problem, Propositions \ref{nn-12} and \ref{mdsnkd} indicate that a nontrivial protocol does not exist for $n=3,4$.

\section*{Acknowledgements}

The research of Chong Shangguan and Itzhak Tamo was supported by ISF grant No. 1030/15 and NSF-BSF grant No. 2015814. The authors would like to thank the three anonymous reviewers for their comments which are very helpful to the improvement of this paper.

\bibliographystyle{plain}
\bibliography{testing_equality}

 \end{document}